\documentclass[a4paper]{amsproc}
\usepackage{amssymb}
\usepackage{amscd}
\usepackage[dvips]{graphicx}

\theoremstyle{plain}
 \newtheorem{thm}{Theorem}[section]
 \newtheorem{prop}{Proposition}[section]
 \newtheorem{lem}{Lemma}[section]
 
\theoremstyle{definition}
 
 \newtheorem{rem}{Remark}[section]

\numberwithin{equation}{section}

\renewcommand{\leq}{\leqslant} 
\renewcommand{\setminus}{\smallsetminus}

\def\cj {\char'017}

\def\ji {\char'032}
\def\ja {\char'037}
\def\m  {\char'176}

\font\srit=wncyi8
 \font\srrm=wncyr8

\newcommand{\R}{\mathbb{R}}

\DeclareMathOperator{\diag}{\mathrm{diag}}

\DeclareMathOperator{\Hol}{\mathrm{Hol}}

\title[the Maupertuis principle]{ON THE PRINCIPLE OF STATIONARY ISOENERGETIC ACTION} 
\author[Jovanovi\'c]{\bfseries Bo\v zidar Jovanovi\'c}

\address{
Mathematical Institute SANU \\
Serbian Academy of Sciences and Arts \\
Kneza Mihaila 36, 11000 Belgrade\\
Serbia}
\email{bozaj@mi.sanu.ac.rs}

\begin{document}

\begin{abstract}
We present several variants of the Maupertuis principle, both on
the exact and the non exact symplectic manifolds.
\end{abstract}

\maketitle

\centerline{\it Dedicated to the memory of Academician Anton
Bilimovi\' c  (1879-1970)}

\tableofcontents

\section{Introduction}

\subsection{}
The principle of least action, or the principle of stationary
action, says that the trajectories of a mechanical system can be
obtained as extremals of a certain action functional. It is one of
the basic tools in physics being applied both in classical and
quantum setting.

Consider a Lagrangian system $(Q,L)$, where $Q$ is a configuration
space and $L(q,\dot q,t)$ is a Lagrangian, $L: TQ \times \R\to
\R$. Let $q=(q_1,\dots,q_n)$ be local coordinates on $Q$. The
motion of the system is described by the Euler--Lagrange equations
\begin{equation} \label{Lagrange}
\frac{d}{dt}\frac{\partial L}{\partial \dot q_i}=\frac{\partial
L}{\partial q_i}, \quad i=1,\dots,n.
\end{equation}

The solutions of the Euler--Lagrange equations are exactly the
critical points of the action integral
$$
S_L({\gamma})=\int_a^b L(q,\dot q,t)dt
$$ in a class of curves ${\gamma}:
[a,b]\to Q$ with fixed endpoints ${\gamma}(a)=q_0$,
${\gamma}(b)=q_1$ (the {\it Hamiltonian principle of least action}
(1834), e.g., see \cite{Po}).

The Legendre transformation $\mathbb FL: TQ\to T^*Q $ is defined
by
\begin{equation}\label{legendre}
\mathbb FL(q,\xi,t)\cdot \eta=
\frac{d}{ds}\vert_{s=0}L(q,\xi+s\eta,t) \quad \Longleftrightarrow
\quad p_i=\frac{\partial L}{\partial \dot q_i}, \quad i=1,\dots,
n,
\end{equation}
where $\xi,\eta\in T_q Q$ and $(q_1,\dots,q_n, p_1,\dots,p_n)$ are
canonical coordinates of the cotangent bundle $T^*Q$. In order to
have a Hamiltonian description of the dynamics (see the section
below),  we suppose that the Legendre transformation
\eqref{legendre} is a diffeomorphism. The corresponding Lagrangian
$L$ is called {\it hyperregular} \cite{MR}.

If the Lagrangian $L$ does not depend on time then  the equations
\eqref{Lagrange} possess the energy first integral
\begin{equation}\label{Energy}
E(q,\dot q)=\mathbb FL(q,\dot q)\cdot \dot q-L(q,\dot
q)=\sum_{i}\frac{\partial L}{\partial \dot q_i}\dot q_i-L.
\end{equation}

In that case we have

\begin{thm}[the Maupertuis principle] Suppose that $h$ is a regular value of $E$.
Among all curves $q={\gamma}(\tau)$ connecting two points $q_0$
and $q_1$ and parametrized so that the energy has a fixed value
$E=h$, the trajectory of the equations of dynamics
\eqref{Lagrange} is an extremal of the reduced action
\begin{equation}\label{red-action}
S({\gamma})=\int_a^b \mathbb
FL\left(q(\tau),\frac{dq}{d\tau}\right)\cdot \frac{dq}{d\tau}\,
d\tau=\int_a^b \frac{\partial L}{\partial\dot q}(\tau)\cdot
\frac{dq}{d\tau}\, d\tau, \quad q_0={\gamma}(a), \,
q_1={\gamma}(b)
\end{equation}
\end{thm}

It is important to note that the interval $[a,b]$, parametrizing
the curve $q={\gamma}(\tau)$, is not fixed and it can be different
for different curves being compared, while the energy must be the
same.

Contrary to the Hamiltonian principle, the {\it Maupertuis
principle}, or {\it principle of stationary isoeneretic action}
determines the shape of a trajectory but not the time. In order to
determine the time, we have to use the energy constant.

Historically, a variant of Theorem 1.2 was the first variational
approach to mechanics. It is attributed to Maupertuis (1744),
Euler (1744) and Jacobi (1842), who gave an important geometric
interpretation of the principle (see \cite{Po}).

\subsection{}

The classical proofs of the Maupertuis principle can be found in
\cite{Po, Wh, Ar}. In Serbian, see the second volume of
Bilimovi{\' c}'s course in Theoretical mechanics \cite{Bi}, or
Dragovi\'c and Milinkovi\'c's monograph \cite{DM}.

Weinstein \cite{W1} and Novikov \cite{N} formulated multi-valued
variational principles that provided the study of the existence of
periodic orbits on non exact symplectic manifolds. We feel a need
to present these results, along with the classical ones, in a
unified way.

In the first part of the paper, we derive the principle of
stationary isoenergetic action, both on the exact (Section 2) and
 the non-exact symplectic manifolds (Section 3). The variants of
the Maupertuis principle presented in Section 3 are our small
contribution to the subject. They slightly differ from the
existing variational principles formulated either for closed
trajectories, or formulated without imposing the constraint given
by the energy.

In the second part of the paper we point out a contact
interpretation of the Maupertuis principle (Sections 4, 5). There,
it is illustrated how some of the well known properties of the
system of harmonic oscillators, the Kepler  problem (Moser's
regularization) and the Neumann system (relationship with a
geodesic flow on an ellipsoid), have natural descriptions within a
framework of the contact geometry. We believe that one should
expect other interesting relations between the contact structures
and integrable systems as well.

It is a great pleasure to dedicate this paper to Anton
Bilimovi\'c, since his work has fundamentally influenced the
development of Serbian theoretical mechanics.

\section{Principle of stationary isoeneretic action in a phase space}

\subsection{Hamiltonian equations}
Let $L(q,\dot q,t)$ be a hyperregular Lagrangian. We can pass from
velocities $\dot q_i$ to the momenta $p_j$ by using the standard
Legendre transformation \eqref{legendre}. In the coordinates
$(q,p)$ of the cotangent bundle $T^*Q$, the equations of motion
\eqref{Lagrange} read:
\begin{equation} \label{1}
\frac{dq_i}{dt}=\frac{\partial H}{\partial p_i},\qquad
\frac{dp_i}{dt}=-\frac{\partial H}{\partial q_i}, \qquad
i=1,\dots,n,
\end{equation}
where the Hamiltonian function $H(q,p,t)$ is the {\it Legendre
transformation} of  $L$
\begin{equation*} \label{haml}
 H(q,p,t)=E(q,\dot q,t)\vert_{\dot q=\mathbb FL^{-1}(q,p,t)}=\mathbb FL(q,\dot q,t)\cdot \dot q-L(q,\dot
q,t)\vert_{\dot q=\mathbb FL^{-1}(q,p,t)}.
\end{equation*}

Let $p\,dq=\sum_i p_i dq_i$ be the {\it canonical 1-form} and
$$
\omega=d(p\,dq)=dp \wedge dq=\sum_{i=1}^n dp_i\wedge dq_i
$$
the
{\it canonical symplectic} form of the cotangent bundle $T^*Q$.
The system of equations \eqref{1} is Hamiltonian, that is the
vector field
$$
X_H(q,p)=(\partial H/\partial p_1,\dots,\partial H/\partial
p_n,-\partial H/\partial q_1,\dots,-\partial H/\partial q_n)
$$
can be defined by
\begin{equation}\label{Hamiltonian}
i_{X_H}\omega (\,\cdot\,)=\omega(X_H,\,\cdot\,)=-dH(\,\cdot\,).
\end{equation}

\subsection{Characteristic line bundles}
More generally, consider a $2n$-dimensional symplectic manifold
$P$ with a closed, non-degenerate 2-form $\omega$. Let $H:
P\times\R\to \R$ be a smooth, in general time dependent, function.
Consider the corresponding Hamiltonian equation
\begin{equation}\label{2}
\dot x=X_H,
\end{equation}
where the Hamiltonain vector field $X_H(x,t)$ is defined by
\eqref{Hamiltonian}.

If the Hamiltonian $H$ does not depend on time, it is the first
integral of the system. Let $M$ be a regular connected component
of the invariant variety $H=h$, which means $dH\vert_M \ne 0$.

Since $dH(\xi)=0$, $\xi\in T_x M$, from \eqref{Hamiltonian} we see
that $X_H$ generates the symplectic orthogonal of $T_x M$ for all
$x\in M$ --- the {\it characteristic line bundle} $\mathcal L_M$
of $M$. It is the kernel of the form $\omega$ restricted to $M$:
$$
\mathcal L_M=\{\xi\in T_x M\, \vert \, \omega(\xi,T_x M)=0, \,
x\in M\}. $$

Note that $\mathcal L_M$ is determined only by $M$ and not by $H$.
If $F$ is another Hamiltonian defining $M$, $M\subset F^{-1}(c)$,
$dF\vert_M\ne 0$, then the restrictions of the Hamiltonian vector
fields $X_H$ and $X_F$ to $M$ are proportional.

A {\it variation} of a curve $\gamma: [a,b]\to M$ is a mapping:
$\Gamma: [a,b]\times [0,\epsilon] \to M$, such that
$\gamma(t)=\Gamma(t,0)$, $t\in [a,b]$. Denote
$\gamma_s(t)=\Gamma(t,s)$ and
$\delta\gamma(t)=\frac{d}{ds}\vert_{s=0}\gamma_s(t)\in
T_{\gamma(t)} M$.

From Cartan's formula we get  (e.g., see Griffits \cite{Gr}):

\begin{lem} Let $(M,\alpha)$ be a manifold endowed with a 1-form
$\alpha$, $\gamma: [a,b]\to M$ be an immersed curve and $\Gamma$
be a variation of $\gamma$. The Lie derivative of the form
$\Gamma^*\alpha$ in the direction of $\partial/\partial s$ at the
points $[a,b]\times \{0\}$ is equal to
$$
L_{\partial/\partial
s}\Gamma^*\alpha\vert_{(t,0)}=\gamma^*(i_{\delta\gamma(t)}
d\alpha)+d\gamma^*(\alpha(\delta\gamma(t))).
$$
\end{lem}

\begin{thm} Assume that the symplectic form $\omega$ is exact: $\omega=d\alpha$.
Let $M$ be a regular component of the invariant hypersurface
$H^{-1}(h)$. The integral curves $\gamma: [a,b]\to M$ of the
characteristic line bundle $\mathcal L_M$ are extremals of the
(reduced) action functional
$$
A(\gamma)=\int_\gamma \alpha=\int_a^b \alpha(\dot\gamma)dt
$$
in the class of variations $\gamma_s(t)$ such that
$\alpha(\delta\gamma(a))=\alpha(\delta\gamma(b))=0$.
\end{thm}

The proof is a direct consequence of Lemma 2.1. We have
\begin{equation}\label{izvod}
\frac{d}{ds}\left(\int_{\gamma_s}
\alpha\right)\vert_{s=0}=\int_a^b
\omega(\delta\gamma(t),\dot\gamma(t))dt+\alpha(\delta\gamma(b))-\alpha(\delta\gamma(a)).
\end{equation}

The expression above is equal to zero for all variations
$\gamma_s(t)$ if and only if $\dot \gamma$ is in the kernel of the
form $\omega=d\alpha$ restricted to $M$. That is, $\gamma(t)$ is
an integral curve of the line bundle $\mathcal L_M$.

\subsection{}
Applying Theorem 2.1 to the symplectic space $(T^*Q,dp\wedge dq)$
we obtain Poincar\'e's formulation of the Maupertuis principle in
a phase space \cite{P}.

\begin{thm}
If the Hamiltonian function $H=H(q,p)$ does not depend on time,
then the phase trajectories of the canonical equations \eqref{1}
lying on the regular connected component $M$ of the surface
$\{H(q,p)=h\}$ are extremals of the reduced action
\begin{equation}\label{red-act-ham}
A(\gamma)=\int_\gamma p\,dq
\end{equation}
 in the class of curves $\gamma$ lying on
$M$ and connecting the subspaces $T^*_{q_0} Q$ and $T^*_{q_1} Q$.
\end{thm}

Note that Theorem 1.1 follows from Theorem 2.2 (e.g., see Arnold
\cite{Ar}). Suppose that the Hamiltonian system \eqref{1} is a
Legendre transformation of the Lagrangian system \eqref{Lagrange}.
The main observation is that if $\underline{\gamma}(\tau)$ is a
configuration space curve parametrized such that
$E(\underline{\gamma},d\underline{\gamma}/d\tau)=h$, then the
lifted curve $ \gamma=\mathbb
FL(\underline{\gamma},d\underline{\gamma}/{d\tau}) $ lyes on $M$
and the reduced actions \eqref{red-action} and \eqref{red-act-ham}
for $\underline{\gamma}$ and $\gamma$ are equal:
$S(\underline{\gamma})=A(\gamma)$ (see Fig. 1).

\begin{figure}[ht]
\includegraphics{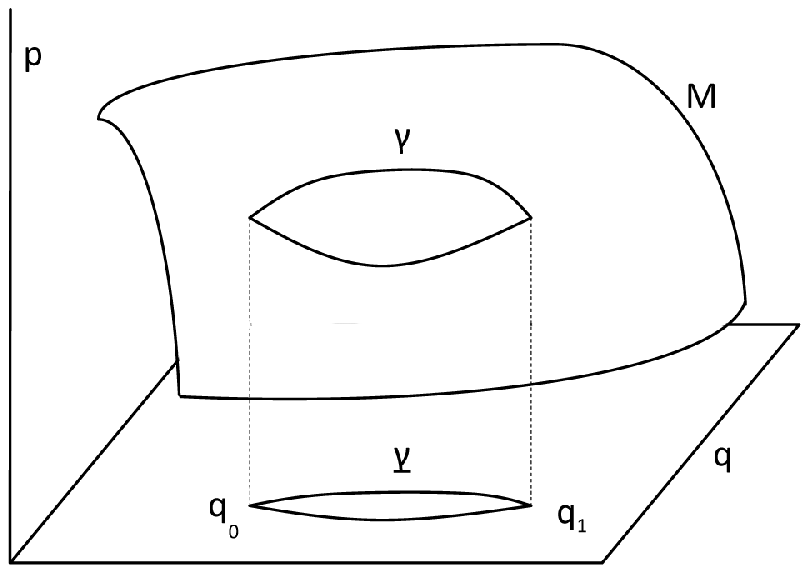}
\caption{}
\end{figure}

\subsection{Jacobi's metric}

Consider a natural mechanical system on $Q$ defined by the
Lagrangian function:
\begin{equation}
L(q,\dot q)=T+B-V=\frac12\sum_{ij}K_{ij}\dot q_i\dot q_j+\sum_i
B_i \dot q_i-V(q). \label{LG}
\end{equation}
Here $ds^2=\sum_{ij}K_{ij}dq_idq_j$ is a Riemannian metric on $Q$,
$V(q)$ is a potential function and $\theta=\sum_i B_i dq_i$ is a
1-form defining a gyroscopic (or magnetic) field $\sigma=d\theta$
(see Section 3).

The energy of the system \eqref{Energy} is the sum of the kinetic
and the potential energy
$$
E(q,\dot  q)=T+V=\frac12\sum_{ij}K_{ij}\dot q_i\dot q_j+V(q).
$$

In the region of the configuration space $Q_h$ where $V(q)<h$, we
can define the {\it Jacobi metric}
\begin{equation}\label{Jacobi}
ds^2_J=2{(h-V(q))}\,ds^2=2{(h-V(q))} \,\sum_{ij} K_{ij}dq_idq_j.
\end{equation}

The following version of the Maupertuis principle for Lagrangians
of the form \eqref{LG} is well known (e.g., see Kozlov \cite{Ko}).

\begin{thm}\label{JM}
Among all curves $q={\gamma}(\tau)$ connecting the points $q_0,
q_1\in Q_h$ and parametrized so that the energy has a fixed value
$E=h$, the trajectory of the equations of dynamics
\eqref{Lagrange} with Lagrangian \eqref{LG} is an extremal of the
integral
\begin{equation}\label{ra}
S({\gamma})=\int_{{\gamma}} ds_J+\theta.
\end{equation}
In particular, if there are no gyroscopic forces, the trajectories
of the system within $Q_h$, up to reparametrization, are geodesic
lines of the Jacobi metric $ds^2_J$.
\end{thm}

Indeed, in order to guarantee a fixed value of the energy
$$
E=T+V=\frac12\sum_{ij}K_{ij}\frac{dq_i}{d\tau}\frac{dq_j}{d\tau}+V(q)=\frac12\left(\frac{ds}{d\tau}\right)^2+V(q)=h,
$$ the
parameter $\tau$ of the curve $q={\gamma}(\tau)$ must be
proportional to the length $d\tau=ds/\sqrt{2(h-V)}$. Therefore
\begin{eqnarray*}
\int_a^b \frac{\partial L}{\partial\dot q}(\tau)\cdot
\frac{dq}{d\tau}\, d\tau &=& \int_a^b \left(\sum_{ij} K_{ij}
\frac{dq_i}{d\tau}\frac{dq_j}{d\tau}+\sum_i
B_i\frac{dq_i}{d\tau}\right) d\tau \\
&=& \int_a^b \left(2(h-V(q))+\sum_i
B_i\frac{dq_i}{d\tau}\right)d\tau=\int_{{\gamma}} ds_J+\theta.
\end{eqnarray*}

\begin{rem}
The variational principle stated in Theorem \ref{JM} is used in
the study of periodic trajectories of natural mechanical systems
with exact magnetic fields (see \cite{Ta} and references therein).
Note also that the Maupertuis principle for a configuration space
$Q$ being a Banach space can be found in \cite{MR, RBB}.
\end{rem}

\subsection{The Hamiltonian principle of least action.}

Consider a {\it Poincar\'e--Cartan} 1-form $pdq-Hdt$ on the
extended phase space $T^*Q \times\R(q,p,t)$, where $H: T^*Q
\times\R \to \R$ is a Hamiltonian function. The phase trajectories
of the canonical equations \eqref{1} are extremals of the action
\begin{equation}\label{poincare}
A_H(\gamma)=\int_\gamma pdq-Hdt
\end{equation}
in the class of curves $\gamma(t)=(q(t),p(t),t)$ connecting the
subspaces $T^*_{q_0} Q\times \{t_0\} $ and $T^*_{q_1} Q\times
\{t_1\}$ (Poincar\'e's modification of the Hamiltonian principle
of least action \cite{P}). Namely, a vector $(\xi,1)$, $\xi\in
T_{(q,p)}(T^*Q)$ belongs to $\ker d(pdq-Hdt)$ at $(q,p,t)$ if and
only if $\xi=X_H(q,p,t)$ (see \cite{Ar, MR}).

Obviously, we can replace $(T^*Q,dp\wedge dq)$ by an arbitrary
exact symplectic manifold $(P,\omega=d\alpha)$. In particular, if
we consider the action $A_H(\gamma)=\int_\gamma \alpha-Hdt$ on the
free loop space $\Omega(P)=C^\infty(S^1,P)$, $S^1 =\R/\mathbb Z$
of $P$ and $H$ is 1-periodic in $t$-variable, then the critical
points of $A_H$ are 1-periodic orbits of the equation \eqref{2}.

For a given time-independent Hamiltonian $H: P \to\R$ with a
regular level set $H^{-1}(h)$, the periodic orbits having all
positive periods and energy $h$ can be obtained by the use of
modified action:
\begin{equation}\label{Rab}
A_{H,h}(\gamma,\lambda)=\int_0^1 \alpha(\dot
\gamma)dt-\lambda\int_0^1 (H(\gamma(t))-h)dt,
\end{equation}
defined on the space $\Omega(P)\times\R^+$ (see \cite{Ra, W1}).
The critical points $(\gamma,\lambda)$ of $A_{H,h}$ correspond to
$\lambda$-periodic orbits $x(t)=\gamma(t/\lambda)$ that lie on the
energy hypersurface $H^{-1}(h)$. Moreover, Weinstein defined
actions $A_H$ and $A_{H,h}$ when the symplectic form is not exact
as well \cite{We}.

The Lagrangian analogue of the functional \eqref{Rab} is
$$
S_{L,h}({\gamma},\lambda)=\int_0^1 \lambda
L({\gamma},\dot{{\gamma}}/\lambda)dt+\lambda h, \qquad
{\gamma}\in\Omega(Q), \, \lambda>0
$$
(see \cite{CMP}). The pair $({\gamma},\lambda)$ is a critical
point of $S_{L,h}$ if and only if $q(t)={\gamma}(t/\lambda)$ is a
$\lambda$-periodic solution of the Euler--Lagrange equation
\eqref{Lagrange} with energy $h$.

Variational principles related to the  action \eqref{poincare},
which arise by a reduction process are given in
 \cite{CMPR}.

\section{The Maupertuis principle on non-exact symplectic
manifolds}

\subsection{Magnetic flows}
Consider a natural mechanical system given by the Lagrangian
function \eqref{LG}. After the Legendre transformation, it takes
the form \eqref{1} with the Hamiltonian function
\begin{equation}\label{HG}
H(q,p)=\frac12 \langle p-\theta,p-\theta \rangle +V(q)=\frac12
\sum_{ij} K^{ij}(p_i-B_i)(p_j-B_j)+V(q),
\end{equation}
where $K^{ij}$ is the inverse of the metric tensor $K_{ij}$.

The transformation
\begin{equation}\label{T}
T_\theta: (q,p)\longmapsto (q,p-\theta)
\end{equation}
is a symplectomorphism between $(T^*Q,dp\wedge dq)$ and a
"twisted" cotangent bundle $(T^*Q, dp\wedge dq+\pi^*\sigma)$,
where $\pi: T^*Q\to Q$ is the natural projection and
$\sigma=d\theta$.

In new coordinates, also denoted by $(q,p)$, the Hamiltonian
\eqref{HG} takes the usual form, the sum of the kinetic and the
potential energy:
$$
H(q,p)=\frac12 \langle p,p\rangle+V(q)=\frac12\sum_{ij} K^{ij} p_i
p_j+V(q),
$$
while the equations of motion take "non-canonical" form:
\begin{equation} \label{magnetic_flow}
\frac{dq^i}{dt}=\frac{\partial H}{\partial p_i}, \qquad
\frac{dp_i}{dt}=-\frac{\partial H}{\partial
q^i}+\sum_{j=1}^{n}F_{ij} \frac{\partial H}{\partial p_j},
\end{equation}
where $\sigma=\sum_{1 \leq i<j \leq n} F_{ij}(q) dq_i \wedge
dq_j$. The equations are Hamiltonian with respect to the
symplectic form $\omega=dp\wedge dq+\pi^*\sigma$.

One can consider the system \eqref{magnetic_flow} associated to a
non-exact 2-form $\sigma$ as well (for example, the motion of a
particle in a magnetic monopole field \cite{MR}). In this case,
the Lagrangian \eqref{LG} is defined only locally. Nevertheless,
it is very interesting that the Hamiltonian (Weinstein \cite{W1}
and Tuynman \cite{Tu}) and the Maupertuis principles (Novikov
\cite{No}) of least action can be still defined.

\subsection{Multivalued reduced action}
Let $(P,\omega)$ be a non exact symplectic manifold and let
$M=H^{-1}(h)$ be a regular isoenergetic hypersurface. The main
observation concerning the Maupertuis principle can be stated as
follows (see \cite{No, Ko} for the reduced action \eqref{ra}).

 Let $U\subset P$ be a region where $\omega$ is exact and let
$\omega=d\alpha_1=d\alpha_2$. Consider a variation
$\gamma_s(t)=\Gamma(t,s)$, $t\in [0,1], s\in [0,\epsilon]$ with
fixed endpoints of a curve $\gamma$ lying in $M\cap U$. Then
$$
\int_{\gamma_\epsilon}
\alpha_1-\int_{\gamma_0}\alpha_1=-\int_{[0,1]\times [0,\epsilon]}
\Gamma^* d\alpha_1 =-\int_{[0,1]\times [0,\epsilon]} \Gamma^*
d\alpha_2=\int_{\gamma_\epsilon} \alpha_2-\int_{\gamma_0}\alpha_2.
$$
Therefore
$$
\int_{\gamma_\epsilon} \alpha_1-\int_{\gamma_0} \alpha_1=
\int_{\gamma_\epsilon} \alpha_2-\int_{\gamma_0} \alpha_2
$$
and, although $\int_\gamma \alpha_i$ depends on the form
$\alpha_i$,  the derivative
$$
\frac{d}{ds}\vert_{s=0}\int_0^1 \alpha(\dot \gamma_s)dt =\int_0^1
\omega(\delta\gamma(t),\dot\gamma(t))dt
$$ does not depend on $\alpha_i$,  $i=1,2$.
One can define an appropriate multi-valued functional on a space
of paths with fixed endpoints, such that an extremal (if exist) is
exactly the integral curve of the characteristic foliation on $M$.
However, as in the case of the symplectic homology (see
\cite{HZ}),  the situation simplifies in the aspherical case which
is considered below.

\subsection{Aspherical symplectic manifolds}
The symplecic manifold $(P,\omega)$ is {\it aspherical} if
$\omega$ vanishes on $\pi_2(P)$. Ofcourse, if $\omega$ is exact or
$\pi_2(P)=0$, then $(P,\omega)$ is aspherical.

Consider the equation \eqref{2}, where $H$ does not depend on
time. Let $M$ be a regular component of $H^{-1}(h)$ and $c:
[0,1]\to P$ be an immersed curve with endpoints $x_0=c(0)\in M$
and $x_1=c(1)\in M$. Define $\Omega_{c}^h(x_0,x_1)$ as the space
of regular paths that are homotopic to $\tilde\gamma$ in $P$:
$$
\Omega_{c}^h(x_0,x_1)=\{\gamma: [0,1]\to M \,\vert\,
\gamma(0)=x_0,\, \gamma(1)=x_1, \dot\gamma(t)\ne 0, t\in [0,1], \,
\gamma \sim_P c\}.
$$

The space of all regular paths connecting $x_0$ and $x_1$ and
laying in $M$ is the union
\begin{equation*}\label{putevi}
\Omega^h(x_0,x_1)=\bigcup_c \Omega_{c}^h(x_0,x_1),
\end{equation*}
where we take representatives $c$ for all non-homotopic paths (in
$P$) connecting $x_0$ and $x_1$.

If we suppose that $(P,\omega)$ is simplectically aspherical then
we can define a single-valued {\it reduced action}:
\begin{equation}\label{action}
A: \Omega^h(x_0,x_1) \to \R, \qquad
A(\gamma)\vert_{\Omega_{c}^h(x_0,x_1)}=\int_D f_\gamma^*\omega,
\end{equation}
where
$$
D=\{z\, \vert\, z\in\mathbb C,\, \vert z\vert \le 1\}
$$
is the unit disk, $f_\gamma: D\to P$ is an arbitrary mapping that
is smooth for $\vert z \vert <1$, continuous on $D$  and
$\gamma(t)=f_\gamma(\exp(\sqrt{-1}\pi t))$,
$c(t)=f_\gamma(\exp(\sqrt{-1}\pi (2-t))$, $t\in [0,1]$. That is,
$f(D)$ is a surface with the boundary $\partial D=\gamma \cdot
c^{-1}$.

\begin{figure}[ht]
\includegraphics{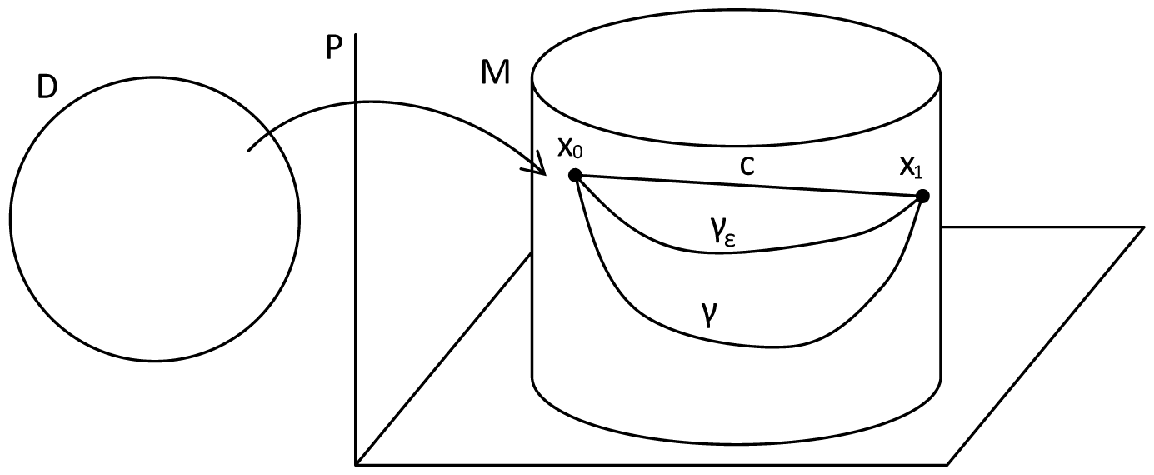}
\caption{}
\end{figure}

Since $\gamma\sim_P c$ we can always find a mapping $f$ with
required properties. From $\omega\vert_{\pi_2(P)}=0$, the value
$A(\gamma)$ does not depend on the choice of $f$.

\begin{thm}
The integral curves $\gamma: [0,1]\to M$ of the characteristic
line bundle $\mathcal L_M$ that connect $x_0$ and $x_1$ are
extremals of the reduced action \eqref{action}.
\end{thm}

\begin{proof}
Consider a variation $\gamma_s(t)=\Gamma(t,s)$,  $t\in [0,1]$,
$s\in [0,\epsilon]$ of $\gamma$ lying in $M$. By using
$\omega\vert_{\pi_2(P)}=0$, we get
$$
A(\gamma_\epsilon)-A(\gamma)=\int_{D}
(f^*_{\gamma_\epsilon}\omega-f^*_\gamma\omega)= -\int_{[0,1]\times
[0,s]} \Gamma^*\omega=\int_0^1\int_0^\epsilon
\omega(\frac{\partial\Gamma}{\partial
s},\frac{\partial\Gamma}{\partial t}) dt ds.
$$
Thus, as above,
$$
\frac{d}{ds}\vert_{s=0}A(\gamma_s) =\int_0^1
\omega(\delta\gamma(t),\dot\gamma(t))dt,
$$
is zero for all variations $\gamma_s(t)$ if and only if the
velocity vector field $\dot\gamma(t)$ is a section of
$\ker\omega\vert_M$.
\end{proof}

\subsection{A torus valued reduced action}

 Tuynaman proposed a torus-valued action, such
 that the multi-valued Poincar\'e action \eqref{poincare} can be seen as a composition of a multi-valued function
on a torus and a torus-valued action \cite{Tu}. In this subsection
we follow Tuynman's construction \cite{Tu} in order to formulate
the principle of stationary isoenergetic action.

Consider a manifold $P$ with a symplectic 2-form
$$
\omega=\sum_{a=1}^n \mu_a \beta^a,
$$
where $\beta^a$ are 2-forms, representing integrals cohomology
classes. We take decomposition with minimal $n$. Then the
parameters $\mu_a$ are independent over $\mathbb Q$, in particular
$\mu=\mu_1+\dots+\mu_n\ne 0$. To $\omega$ we associate the 1-form
$$
\lambda=\sum_{a=1}^n \mu_a dy^a
$$
on a torus $\mathbb T^{n}=\{(\exp({\sqrt{-1}
y^1}),\dots,\exp(\sqrt{-1} y^n)\}$. It can be consider as a
differential of a multi-valued function $\Lambda$ on $\mathbb
T^n$: $\lambda=d\Lambda$.  Also, for $a=1,\dots,n$, let us
 define principal $S^1$-bundles
\begin{equation} \label{BUNDLE}
\begin{aligned}
S^1\;\longrightarrow\;\;& Y_a \\
&\;\Big\downarrow{}^{\rho_a} \\
&\;P
\end{aligned}
\end{equation}
having the connections $\theta^a$ with the curvature forms
$\beta^a$ (see Kobayashi \cite{Kob}).

Let $\gamma(t)$, $t\in [t_0,t_1]$ be a piece-wise smooth, closed
curve on $P$. Recall, a piece-wise smooth curve $\tilde
\gamma^a(t) \subset Y_a$ is a {\it horizontal lift} of $\gamma$ if
$\rho_a\circ\tilde\gamma^a(t)=\gamma(t)$ and
$\theta^a(\frac{d}{dt} \tilde\gamma^a(t))=0$, whenever the
velocity vector is defined. The {\it holonomy} $\Hol^a(\gamma)$ is
an element $g\in S^1$, such that
$g\cdot\tilde\gamma^a(t_0)=\tilde\gamma^a(t_1)$.

\begin{lem}[\cite{Tu}] \label{TU}
 Let $\gamma_s(t)=\Gamma(t,s)$ be a variation of $\gamma: [0,1]\to P$ with fixed
endpoints and let $c: [0,1]\to P$ be an arbitrary curve connecting
$x_0=\gamma(0)$ and $x_1=\gamma(1)$. We have a family of closed
orbits $\bar\gamma_s=\gamma_s\cdot c^{-1}$. The derivative of
$\Hol^a(\bar\gamma_s)$ is given by:
$$
\frac{d\Hol^a(\bar\gamma_s)}{ds}\vert_{s=0}=\int_0^1 \beta^a (\dot
\gamma(t),\delta\gamma(t))dt\,\cdot\,\frac{\partial}{\partial
y^a}.
$$
\end{lem}

Consider the equation \eqref{2}, where $\partial H/\partial t=0$.
Let $M$ be a regular component of $H^{-1}(h)$ and
$\Omega^h(x_0,x_1)$ be a space of regular paths $\gamma: [0,1]\to
M$ that connect points $x_1$ and $x_2$.

For every $\gamma\in\Omega^h(x_0,x_1)$ define a pease-wise smooth,
closed path $\bar\gamma=\gamma\cdot c^{-1}: [0,2]\to M$, where
$c\in\Omega^h(x_0,x_1)$ is fixed. We call
$$ A_{\mathbb T^n}: \,\Omega^h(x_0,x_1)
\longrightarrow \mathbb T^n, \qquad \gamma \longmapsto
(\Hol^1(\bar\gamma),\dots,\Hol^n(\bar\gamma))
$$
a {\it torus valued reduced action}. From Lemma \ref{TU} we have:
$$
\lambda\left(\frac{d}{ds}A_{\mathbb
T^n}(\gamma_s)\vert_{s=0}\right)=\sum_{a=1}^n \mu_a \int_0^1
\beta^a (\dot \gamma(t),\delta\gamma(t))dt= \int_0^1 \omega (\dot
\gamma(t),\delta\gamma(t))dt\,.
$$

Whence, we obtain the following principle of stationary
isoenergetic action on the non-exact symplectic manifolds.

\begin{thm}
A curve $\gamma\in\Omega^h(x_0,x_1)$ is an integral curve of the
characteristic line bundle $\mathcal L_M$ if and only if
$$
\frac{d}{ds}\left(\Lambda \circ A_{\mathbb
T^n}(\gamma_s)\right)\vert_{s=0}=
\lambda\left(\frac{d}{ds}A_{\mathbb
T^n}(\gamma_s)\vert_{s=0}\right)=0
$$
for all variations $\gamma_s\in\Omega^h(x_0,x_1)$.
\end{thm}

For the completeness of the exposition we include:

\begin{proof}[Proof of Lemma \ref{TU}]
In local trivializations
$$
\rho^{-1}(U_i) \cong U_i\times S^1(x_i,y_i\, \mathrm{mod}\,2\pi),
$$
we have local connection 1-forms $\alpha_i$ on $U_i$ such that
$\theta=\alpha_i+dy_i$ (the index $a$ is omitted). The transition
functions between fiber coordinates and connection 1-forms are
given by
\begin{equation}\label{tranzicija}
y_j=y_i+g_{ij}(x), \qquad \alpha_i=\alpha_j+dg_{ij}, \qquad
g^a_{ij}: U_i \cap U_j \to S^1.
\end{equation}
On the other hand, the curvature 2-form is invariant:
$\beta=d\alpha_i=d\alpha_j$.

Suppose $ \gamma_s([t_0,t_1])\subset U_i$, $s\in [0,\epsilon]. $
The local expression for $\tilde\gamma_s$ reads
$$
\tilde\gamma_s(t)=(\gamma_s(t),y_i(t,s)), \qquad
\theta_i\left(\frac{d}{dt}\tilde\gamma(t)\right)=0
\,\,\Longleftrightarrow\,\, \alpha_i(\dot\gamma(t))+\dot
y_i(t,s)=0.
$$
Therefore
\begin{equation}\label{T2}
y_i(t_1,s)=y_i(t_0,s)-\int_{t_0}^{t_1}
\alpha_i(\dot\gamma_s(t))dt.
\end{equation}

By taking the differential of \eqref{T2} at $s=0$ and applying
\eqref{izvod} we get
\begin{equation}\label{T3}
 \delta y_i(t_1)+\alpha_i(\delta\gamma(t_1))=\delta y_i(t_0)+\alpha_i(\delta\gamma(t_0))+\int_{t_0}^{t_1}
\beta(\dot\gamma(t),\delta\gamma(t))dt,
\end{equation}
where $\delta_i y(t)=\frac{d}{ds}y_i(t,s)\vert_{s=0}$,
$\delta\gamma(t)=\frac{d}{ds}\gamma_s(t)\vert_{s=0}$.

Now, assume $t_0<t'_0<t_1<t'_1$, $\gamma_s([t_0,t_1])\subset U_i$
and $\gamma_s([t'_0,t'_1])\subset U_j$. The transformations
\eqref{tranzicija} imply
\begin{equation}\label{T4}
\delta y_i(t)+\alpha_i(\delta\gamma(t))=\delta
y_j(t)+\alpha_j(\delta\gamma(t)), \qquad t\in [t'_0,t_1].
\end{equation}

By combining \eqref{T3} and  \eqref{T4}, it follows
\begin{equation}\label{T6}
 \delta y_j(t'_1)+\alpha_j(\delta\gamma(t'_1))=\delta
y_i(t_0)+\alpha_i(\delta\gamma(t_0))+\int_{t_0}^{t'_1}
\beta(\dot\gamma(t),\delta\gamma(t))dt.
\end{equation}

Let $\bar\gamma_s=\gamma_s\cdot c^{-1}: [0,2]\to M$ and let
$U_1,\dots,U_l$ be local charts, such that
$$
\bar\gamma_s([t_{i-1},t_i])\subset U_i,  \quad
0=t_0<t_1<\dots<t_k=1<t_{k+1}<\dots<t_l=2, \quad s\in
[0,\epsilon].
$$

From the relation \eqref{T6} and  $\delta
\gamma(0)=\delta\gamma(1)=0=\delta \bar\gamma(t)=0$, $t\in [1,2]$,
we get
\begin{equation*}
 \delta \bar y_k(1)-\delta
y_1(0)=\int_{0}^{1} \beta(\dot\gamma(t),\delta\gamma(t))dt, \qquad
\delta \bar y_l(2)-\delta\bar y_k(1)=0.
\end{equation*}
We can suppose that the horizontal lifts of all curves
 start from the same
point in $Y$. Then $\delta y_1(0)=0$. This proves the statement.
\end{proof}

\subsection{Reduced action for magnetic flows}

Let us return to the magnetic equations \eqref{magnetic_flow},
where $H(q,p)$ is an arbitrary smooth function and $\sigma$ is not
exact. Let $M$ be a regular component of $H(q,p)^{-1}(h)$ and let
$\pi: T^*Q\to Q$ be the natural projection.

 As in Theorem
2.2, we not need to fix endpoints in the fiber directions.
 Consider a class
of regular curves $\gamma$ lying on $M$ and connecting the
subspaces $T^*_{q_0} Q$ and $T^*_{q_1} Q$, such that the
projection $\pi(\gamma)$ is homotopic to $c$:
$$
\Omega_c^h(q_0,q_1)=\{\gamma: [0,1]\to M \,\vert\,
\pi(\gamma(0))=q_0,\, \pi(\gamma(1))=q_1, \pi(\gamma) \sim c\},
$$
and a class of all regular paths connecting $T^*_{q_0}Q$ and
$T^*_{q_1}Q$ and laying in $M$:
\begin{equation*}\label{putevi*}
\Omega^h(q_0,q_1)=\bigcup_c \Omega_{c}^h(q_0,q_1),
\end{equation*}
where we take representatives $c: [0,1] \to Q$ for all
non-homotopic paths connecting $q_0$ and $q_1$.

\begin{thm}
Assume $\sigma\vert_{\pi_2(Q)}=0$. The phase trajectories of the
magnetic equations \eqref{magnetic_flow} in the class of curves
$\Omega_c^h(q_0,q_1)$ are extremals of the reduced action
\begin{equation*}
A: \Omega^h(q_0,q_1)\to \R, \quad
A(\gamma)\vert_{\Omega^h_c(q_0,q_1)}=\int_\gamma p\,dq + \int_D
f^*_\gamma \sigma,
\end{equation*}
where $f_\gamma: D\to Q$ is smooth for $\vert z \vert <1$,
continuous on $D$ and
$$
\pi(\gamma(t))=f_\gamma(\exp(\sqrt{-1}\pi t)), \quad
c(t)=f_\gamma(\exp(\sqrt{-1}\pi (2-t)), \quad t\in [0,1].
$$
\end{thm}

If $\sigma\vert_{\pi_2(Q)}\ne 0$, we can use a combination of the
usual reduced action and a torus valued action with respect to the
form $\sigma$. Suppose
$$
\sigma=\sum_{a=1}^n \mu_a \beta^a,
$$
where $\beta^a$ are 2-forms, representing integrals cohomology
classes in $Q$. We take decomposition with minimal $n$. As above,
to $\sigma$ we associate principal $S^1$-bundles $L_a$ over $Q$
having the connections $\theta^a$ with curvature forms $\beta^a$,
$a=1,\dots,n$.

Let us fix  $c: [0,1]\to Q$, $c(0)=q_0$, $c(1)=q_1$. For every
$\gamma\in\Omega^h(q_0,q_1)$, we associate a pease-wise smooth,
closed path $\underline{\gamma}=\pi(\gamma)\cdot c^{-1}: [0,2]\to
Q$. Define
$$
B_{\mathbb T^n}: \,\Omega^h(q_0,q_1) \longrightarrow \mathbb T^n,
\qquad \gamma \longmapsto
(\Hol^1(\underline{\gamma}),\dots,\Hol^n(\underline{\gamma})),
$$
where $\Hol^a$ is the holonomy of the bundle $L_a\to Q$.  Let $
\upsilon=\sum_{a=1}^n \mu_a dy^a $ be a 1-form on $\mathbb T^n$,
considered as a differential of a multi-valued function
$\Upsilon$: $d\Upsilon=\upsilon$.

\begin{thm}
A curve $\gamma\in\Omega^h(q_0,q_1)$ is an integral curve of the
characteristic line bundle $\mathcal L_M$ if and only if
$$
\frac{d}{ds}\left(\int_{\gamma_s} p\,dq-\Upsilon\circ B_{\mathbb
T^n}(\gamma_s)\right)\vert_{s=0}=0
$$
for all variations $\gamma_s\in\Omega^h(q_0,q_1)$.
\end{thm}

\begin{rem}
For various approaches to the existence problem of closed magnetic
orbits, see \cite{CMP, Ta} and references therein. Integrable
magnetic geodesic flows on homogeneous spaces can be found in
\cite{BJ2}.
\end{rem}

\section{Isoenergetic hypersurfaces of contact type}

\subsection{}
A {\it contact form} $\alpha$ on a $(2n+1)$-dimensional manifold
$M$ is a Pfaffian form satisfying $\alpha\wedge(d\alpha)^n\ne0$.
By a {\it contact manifold} $(M,\mathcal H)$ we mean a connected
$(2n+1)$-dimensional manifold $M$ equipped with a nonintegrable
{\it contact} (or {\it horizontal}) {\it distribution} $\mathcal
H$, locally defined by a contact form: $\mathcal
H\vert_U=\ker\alpha\vert_U$, $U$ is an open set in $M$ \cite{LM}.
A contact manifold $(M,\mathcal H)$ is {\it co-oriented (or
stricly) contact} if $\mathcal H$ is defined by a global  contact
form $\alpha$. For a given contact form $\alpha$, the {\it Reeb
vector field} $Z$ is a vector field uniquely defined by
$$
i_Z\alpha=1, \qquad i_Z d\alpha=0.
$$

\subsection{}
In studying  the existence problem of closed Hamiltonian
trajectories on a fixed isoenergetic surface, Weinstein introduced
the following concept \cite{We}. An orientable hypersurface $M$ of
a symplectic manifold $(P,\omega)$ is of {\it contact type} if
there exist a 1-form $\alpha$ on $M$ satisfying
$$
d\alpha=j^*\omega, \quad \alpha(\xi)\ne 0, \, \xi\in \mathcal L_M,
\, \xi\ne 0
$$
where $j: M\to P$ is the inclusion. If $(M,\alpha)$ is of contact
type, since $\mathcal L=\ker\omega_M$, the kernel of $\alpha$
$$
\mathcal H=\{\xi\in T_x M\, \vert\, \alpha(\xi)=0, \, x\in M\}
$$
is a $(2n-2)$-dimensional nonintegrable distribution on which
$d\alpha=\omega$ is nondegenerate. Consequently, $\alpha\wedge
d\alpha^{n-1}$ is a volume form on $M$ and $(M,\mathcal H)$ is a
co-oriented contact manifold.

Now, let $(P,\omega=d\alpha)$ be an exact symplectic manifold.
Consider a regular component $M$ of an isoenergetic surface
$H^{-1}(h)$ ($H$ does not depend on time). If
$\alpha(X_H)\vert_M\ne 0$ then $M$ is of contact type. We say that
$M$ is of {\it contact type with respect to $\alpha$}.

If $M$ is of contact type with respect to $\alpha$ then $\alpha$
has no zeros in some open neighborhood of $M$. Contrary, suppose
that an 1-form $\alpha$ has no zeros in some open neighborhood of
$M$. Then, from the nondegeneracy of $\omega$,  there exists a
unique vector field $E$ such that

\begin{equation}\label{liouville}
i_{E} \omega=\alpha.
\end{equation}
The vector field $E$ has no zeros. From Cartan's formula, the
condition $i_E \omega=\alpha$ is equivalent to $ L_E\omega=\omega,
$ i.e., $E$ is the {\it Liouville vector field} of $\omega$. We
have (e.g., see Libermann and Marle \cite{LM}):

\begin{lem}
A regular connected component $M$ of an isoenergetic surface
$H^{-1}(h)$ is of contact type with respect to $\alpha$ if and
only if the Liouville vector field defined by \eqref{liouville} is
transverse to $M$.
\end{lem}

\begin{proof}
Since  $i_E\omega^{n}=n\alpha\wedge d\alpha^{n-1}$, the kernel of
$\alpha\wedge d\alpha^{n-1}$ is the vector bundle generated by
$E$. Therefore $\alpha\wedge d\alpha^{n-1}\vert_M$ is a volume
form on $M$ at $x$ if and only if $E(x)\notin T_x M$.
\end{proof}

Let $M$ be of contact type with respect to $\alpha$ and let $Z$ be
the corresponding Reeb vector field on $M$:
$$
i_Z d\alpha\vert_M =0, \qquad \alpha(Z)=1.
$$

Since $Z$ is a section of $\ker d\alpha\vert_M$, it is
proportional to $X_H\vert_M$: $Z=\mathcal N X_H\vert_M$, $\mathcal
N\ne 0$. Consequently, the flow of $Z$ can be seen as a flow of
$X_H\vert_M$ after a time reparametrization $dt=\mathcal N d\tau$:
\begin{equation}\label{TR1}
\frac{d x}{d\tau}=\frac{dx}{dt}\frac{dt}{d\tau}=X_H(x)\cdot
\mathcal N(x)=Z(x), \qquad x\in M.
\end{equation}
Alternatively, we can change the Hamiltonian $H$. Extend $\mathcal
N$ to a neighborhood of $M$. Then
\begin{equation}\label{TR2}
X_{\mathcal N(H-h)}(x)=\mathcal N(x) X_H(x), \qquad x\in M.
\end{equation}

Based on the observations \eqref{TR1}, \eqref{TR2}, we have the
following statement.

\begin{lem}\label{jakobijeva}
The function
$$
H_0=\frac{H-h}{E(H)}
$$
 has $M$ as an
invariant surface and the Hamiltonian vector field
$X_{H_0}\vert_M$ is equal to the Reeb field $Z$. If $\rho$ is any
smooth function of a real variable, such that $\rho'(\lambda)=1$,
then $\rho(H_0+\lambda)$ has the same property. In particular, for
$\rho(x)=-1/(4x)$, $\lambda=-1/2$, we get
\begin{equation}\label{Jacobi-transf}
H_J=\frac{E(H)}{4h-4H+2E(H)}, \quad H_J\vert_M=\frac12, \quad
Z=X_{H_J}\vert_M.
\end{equation}
\end{lem}

\begin{proof}
According to \eqref{Hamiltonian}, \eqref{liouville}, we have
\begin{equation}\label{pomoc}
\alpha(X_F)=\omega(E,X_F)=dF(E)=E(F), \qquad F\in C^\infty(P).
\end{equation}

Thus, $Z=X_H/E(H)\vert_M$, i.e., $\mathcal N=1/E(H)$. It is clear
that $H_0\vert_M=0$, while  \eqref{TR2} implies
$Z=X_{H_0}\vert_M$.

Let $\rho$ is a smooth function, such that $\rho'(\lambda)=1$.
Then $\rho(H_0+\lambda)\vert_M=\rho(\lambda)$ and
$E(\rho(H_0+\lambda))\vert_M=\rho'(\lambda)E(H_0)=1$.
\end{proof}

\subsection{Exact magnetic flows}
Consider a natural mechanical system given by the Hamiltonian
function \eqref{HG}. The canonical 1-form $pdq$ is different from
zero outside the zero section $\{p=0\}$, where we have the {\it
standard Liouville vector field} $E=\sum_i p_i
{\partial}/{\partial p_i}$ on $T^*Q$.

Since $E(H)=\langle p, p-\theta \rangle$, a regular hypersurface
$M_h=H^{-1}(h)$ is of contact type with respect to $pdq$ within a
region
\begin{eqnarray*}
M_{0,h}&=&\left\{\langle p-\theta,p-\theta \rangle +2V(q)=2h,
\,\langle p, p-\theta \rangle \ne 0 \right\}\\
&=&\left\{\langle p-\theta,p-\theta \rangle +2V(q)=2h, \,\langle
p,p\rangle \ne 2V+\langle\theta,\theta\rangle-2h\right\}
 \subset T^* Q_h.
\end{eqnarray*}

Note that the equation $\langle p,p-\theta\rangle=0\vert_q$,
$\theta_q \ne 0$,  defines an ellipsoid in $T^*_q Q$. Assume
$$
h_*=\max_{q\in Q} \left(V(q)+\frac12 \langle
\theta,\theta\rangle\right) < \infty
$$
(for example, $h_*$ exists if $Q$ is compact). Then
 regular hypersurfaces
$M=H^{-1}(h)$, for $h>h_*$, are of contact type with respect to
$pdq$. The function \eqref{Jacobi-transf} has the form
$$
H_{J}(q,p)={\langle p-\theta,p\rangle }/({4(h-V(q))+2 \langle
\theta,p\rangle}).
$$

In particular, if $\theta \equiv 0$, $H_J$ is the Hamiltonian
function of the geodesic flow of Jacobi's metric \eqref{Jacobi}
and $M_h$ is the corresponding co-sphere bundle over $Q$.

\begin{rem}
The function $\mathcal N$ in the time reparametrization
\eqref{TR1} equals $\mathcal N=1/E(H)$, $E(H)=\langle p, p-\theta
\rangle=2(h-V(q))\vert_M$. That is,
$$
dt=d\tau/2(h-V),
$$
which agrees with Corollary 2.2 (where the time parameter $dt$ of
the original system is denoted by $d\tau$,
$d\tau=ds/\sqrt{2(h-V)}=ds_J/2(h-V)$, and $ds_J$ is the natural
parameter of Jacobi'metric).
\end{rem}

\section{Examples: contact flows and integrable systems}

\subsection{Harmonic oscillators}

Consider the simplest integrable system - the system of $n$
independent harmonic oscillators defined by the Hamiltonian
function
\begin{equation}\label{ham}
H=\sum_i F_i, \qquad F_i=\frac12({a_i} q_i^2+b_i p_i^2), \qquad
1,\dots,n,
\end{equation}
in the standard symplectic linear space $\R^{2n}(q,p)$. Here we
suppose that the products $a_ib_i$, $i=1,\dots,n$ are positive.

By the use of the first integrals $F_i=c_i$, a generic solution of
the equations
\begin{equation}
\dot q_i =b_i{p_i}, \qquad \dot p_i=-a_i q_i, \qquad i=1,\dots,n
\label{h1}
\end{equation}
can be written in the form  $
q_i(t)=\sqrt{\frac{2c_i}{a_i}}\cos\left(\omega_i
t+\varphi_i^0\right)$, $ p_i(t)=-\sqrt{\frac{2c_i}{b_i}}
\sin\left(\omega_i t+\varphi_i^0\right)$,
$\omega_i=\sqrt{a_ib_i}$, where $\varphi^0_i\in [0,2\pi)$ are
determined from the initial conditions. Assume
\begin{eqnarray*}
&&
A_k=a_{r_1+\dots+r_{k-1}+1}=\dots=a_{r_1+\dots+r_{k}},\\
&&B_k=b_{r_1+\dots+r_{k-1}+1}=\dots=b_{r_1+\dots+r_{k}},\\
&& 1\le k \le s, \quad r_1+\dots+r_s=n, \quad r_0=0
\end{eqnarray*}
and that the frequencies
$\sqrt{A_1B_1},\sqrt{A_2B_2},\dots,\sqrt{A_sB_s}$ are independent
over $\mathbb Q$.

Due to the $U(r_1)\times \dots \times U(r_s)$-symmetry, the system
\eqref{h1} has additional Noether integrals
\begin{eqnarray*}\label{integrals}
&&F^k_{ij}=A_k q_iq_j+B_k p_ip_j, \qquad
G^k_{ij}=q_jp_i-p_jq_i, \\
&& {r_1+\dots+r_{k-1}+1} \le i<j\le r_1+\dots+r_{k}, \qquad
 k=1,\dots,s,
\end{eqnarray*}
implying the non-commutative integrability of the system \cite{N,
MF}. Generic trajectories fill up densely invariant
$s$-dimensional invariant isotropic tori generated by the
Hamiltonian vector fields of integrals
$$
H_1=F_1+\dots+F_{r_1}\, , \quad \dots\, ,\quad H_s=
F_{r_1+\dots+r_{s-1}+1}+\dots+F_{r_1+\dots+r_{s}}\,.
$$

The quadric $M_h=H^{-1}(h)$, $h\ne 0$ is of contact type with
respect to the canonical 1-form $p\,dq$ outside $p=0$, where we
have a well defined Jacobi's metric.

However, if instead of $p\,dq$, we take
\begin{equation}\label{pert}
\alpha=\sum_{i=1}^n p_idq_i-\frac12d(\sum_{i=1}^n p_i
q_i)=\frac12\sum_i p_i dq_i-q_i dp_i,
\end{equation}
then $d\alpha=d(p\,dq)=dp\wedge dq$ and the only zero of $\alpha$
is at the origin $0$. The corresponding Liouville vector field is
$$
E=\frac12\sum_i q_i\frac{\partial}{\partial q_i}+p_i
\frac{\partial}{\partial p_i}. 
$$

Since $E(H)=h\vert_{M_h}$, the  quadric $M_h$ is of contact type
with respect to $\alpha$ and the Reeb flow on $M_h$ is
$Z=h^{-1}X_H\vert_{M_h}$.

The above construction  provides natural examples of contact
structures on quadrics within $\R^{2n}$ having the integrable Reeb
flows with $s$-dimensional invariant tori, for any $s=1,\dots,n$.
The case $s=n$ corresponds to contact commutative integrability
introduced by Banyaga and Molino \cite{BM} (see also \cite{KT,
Bo}), while for $s<n$ we have contact noncommutative integrability
recently proposed in \cite{Jo}.

By taking all parameters to be positive ($a_i,b_i>0$,
$i=1,\dots,n$), after rescaling of $M_h$ to a sphere $S^{2n-1}$,
we get $K$-contact structures on a sphere $S^{2n-1}$ given by
Yamazaki (see Example 2.3 in \cite{Y}). In particular, for
$a_1=a_2=\dots=a_n=b_n=1$ we have the {\it standard contact
structure} on a sphere $S^{n-1}=H^{-1}(1/2)$ with the Reeb flow
which defines the Hopf fibration (e.g., see \cite{LM}).

\begin{rem}
A modification of the canonical form $p\,dq$ given by \eqref{pert}
can be applied for starshaped hypersurfaces in $\R^{2n}$. More
generally, consider a regular isoenergetic hypersurface
$M_h=H^{-1}$ in $(T^*Q(q,p),dp\wedge dq)$. It is of contact type
if there exist a closed 1-form $\varphi$ on $M_h$ such that
$p\,dq(X_H\vert_{M_h})+\varphi(X_H\vert_{M_h})\ne 0$. Assume $M_h$
is compact. Then the required 1-form $\varphi$ exists if and only
if $\int_{M_h} p\,dq(X_H) d\mu\ne 0$ for every invariant
probability measure $\mu$ with zero homology (see Apendix B in
\cite{CMP}). In particular, for a compact regular energy surface
$M_h=H^{-1}(h)$ in the standard symplectic linear space
$(R^{2n}(q,p),dq\wedge dq)$ we have the following sufficient
conditions. Suppose:

(i)  $\,\quad$ $p\,dq(X_H)>0$, for $p\ne 0$, $(q,p)\in M$;

(ii) $\quad$ if $M\cap \{p=0\}\ne \emptyset$, then $\frac{\partial
}{\partial q} H(q,0)\ne 0$ at the points $(q,0)\in M$.

Then $M_h$ is of contact type with respect to
$$
\alpha=\sum_{i=1}^n p_idq_i-\epsilon d\left(\sum_{i=1}^n
p_i\frac{\partial}{\partial q_i} H(q,0)\right),
$$
for a certain parameter $\epsilon$  (see \cite{HZ}).
\end{rem}

\subsection{The regularization of Kepler's problem}
The motion of a particle in the central potential filed is
described by the Hamiltonian function
$$
H: \R^{2n}_*=\R^{2n}\setminus \{q=0\}\to \R, \qquad H(q,p)=\frac{
\vert p\vert^2}{2}-\frac{\gamma}{{\vert q\vert}},
$$
where $\langle\cdot,\cdot\rangle$ is the Euclidean scalar product
in $\R^n$. Moser's regularization of Kepler's problem (see
\cite{Mo}) can be interpreted in contact terms as follows.

Let $M_h=\{H=h\}\subset \R^{2n}_*$ be an isoenergetic
hypersurface. Let us interchange the roll of $q$ and $p$ and
consider the form $ \alpha =-\sum_{i=1}^n q_idp_i $ and the
associated Liouville vector field
$$
E=\sum_{i=1}^n q_i\frac{\partial}{\partial q_i}.
$$

Since $ E(H)={\gamma}/{\vert q\vert}$, $M_h$ is of contact type
with respect to $\alpha$. According to Lemma \ref{jakobijeva}, the
Reeb flow on $M_h$ can be seen as a Hamiltonian flow of
$$
H_0=(\vert p\vert^2-2h)\vert q\vert/2\gamma-1.
$$

In order to get a smooth Hamiltonian we can take $F=(H_0+1)^2/2$
(see Lemma \ref{jakobijeva}):
$$
F(q,p)=\frac{(\vert p\vert ^2-2h)^2}{8\gamma^2}\vert q\vert^2.
$$

Then $F\vert_{M_h}=\frac12$, $Z=X_F\vert_{M_h}$ and, moreover,
$X_F$ is defined on the whole $\R^{2n}$.

Assume $h<0$. The Hamiltonian $F(q,p)$ can be interpreted as a
geodesic flow of the metric proportional to
$$
ds^2_h=\frac{dp_1^2+\dots+dp_n^2}{(2h-\vert p \vert^2)^2}.
$$
It represents the round sphere metric obtained by a stereographic
projection (see Moser \cite{Mo}). Thus, for $h<0$, there exist a
compact contact manifold $\bar M_h=M_h \cup S^n$ (a co-sphere
bundle over $S^{n}$) with a Reeb vector field $\bar Z$, which is a
smooth extension of $Z$. In particular, for $n=2$, $\bar M_h \cong
\mathbb{RP}^3$. On $\mathbb{RP}^3$ we have a {\it standard contact
structure}, obtained from the standard contact structure on $S^3$
via antipodal mapping.

Note that for $h>0$, the metric $ds^2_h$ is defined within the
ball of radius $\sqrt{2h}$ and represents Poincar\'e's model of
the Lobachevsky space.

The contact regularization of the restricted 3-body problem is
given in \cite{AFKP}.

\subsection{The Maupertuis
principle and geodesic flows on a sphere} It is well known that
the standard metric on a rotational surface and on an ellipsoid
have the geodesic flows integrable by means of an integral
polynomial in momenta of the first (Clairaut) and the second
degree (Jacobi) \cite{Ar}. A natural question is the existence of
a metric on a sphere $S^2$ with polynomial integral which can not
be reduced to linear or quadratic one. The first examples are
given in \cite{BKF}. Namely, the motion of a rigid body about a
fixed point in the presence of the gravitation field admits
$SO(2)$--reduction (rotations about the direction of gravitational
field) to a natural mechanical system on $S^2$. Starting from the
integrable Kovalevskaya and Goryachev--Chaplygin cases and taking
the corresponding Jacobi's metrics, we get
 the  metrics with additional integrals
of 4-th and 3-th degrees, respectively.

We proceed with a celebrated Neumann system. The Neumann system
describes the motion of a particle on a sphere $\langle
q,q\rangle=1$ with respect to the quadratic potential
$V(q)=\frac12\langle Aq,q\rangle$, $A=\diag(a_1,\dots,a_n)$ (we
assume that $A$ is positive definite). The Hamiltonian of the
system is:
\begin{equation}\label{NK}
H_{N}(q,p)=\frac12\langle p,p\rangle+\frac12\langle Ax,x\rangle\,.
\end{equation}

Here, the cotangent bundle of a sphere $T^*S^{n-1}$ is realized as
a submanifold $P$ of $\R^{2n}$ given by the constraints
\begin{equation}
F_1\equiv \langle q,q\rangle =1, \quad F_2\equiv \langle
q,p\rangle =0. \label{psi}
\end{equation}

The canonical symplectic form on $P\cong T^*S^{n-1}$ is a
restriction of the standard symplectic form $dp \wedge dq$ to $P$.
Let $H:\R^{2n}\to \R$. The Hamiltonian vector field $X_{H}\vert_P$
reads
$$
X_{H}(q,p)\vert_P=X_H(q,p)-\lambda_1 X_{F_1}(q,p)-\lambda_2
X_{F_2} (q,p), \qquad (q,p)\in P,
$$
where the Lagrange multipliers are determined from the condition
that $X_{H}\vert_P$ is tangent to $P$ (e.g., see \cite{Moser}).

There is a well known {Kn\"orrer}'s correspondence between the
trajectories $q(t)$ of the Neumann system \eqref{NK} restricted to
the zero level set of the integral
\begin{equation}\label{MN}
H(q,p)=\frac12\left(\langle A^{-1} q,q\rangle \langle A^{-1}
p,p\rangle-\langle A^{-1}q,p\rangle^2-\langle A^{-1} q,q\rangle
\right).
\end{equation}
and the geodesic lines on an ellipsoid $
E_1^{n-1}=\{x\in\R^n\,\vert\,\langle x, Ax\rangle =1\} $ by the
use of a time reparametrization and the Gauss mapping $q=Ax/\vert
Ax\vert$ \cite{Knorr1}.

Recently, by using optimal control techniques, Jurdjevic obtain a
similar statement for the flow of the system defined by the
Hamiltonian \eqref{MN} \cite{Ju}.

We give the interpretation of Jurdjevic's time change by the use
of Maupertuis principle. Since the potential $V(q)=-\frac12\langle
A^{-1} q,q\rangle$ is negative, the isoenergetic surface
\begin{equation}\label{MNo}
M_0=\{H\vert_P=0\}\subset P \cong T^*S^{n-1}
\end{equation}
is of contact type with respect to $p\,dq\vert_P$. The Reeb vector
field $Z$ equals to the Hamilonian vector field of
\begin{equation}\label{MNJ}
H_J=\frac1{4\langle A^{-1}q,q\rangle}\left(\langle A^{-1}
q,q\rangle \langle A^{-1} p,p\rangle-\langle
A^{-1}q,p\rangle^2\right)\vert_P
\end{equation}
(the Hamiltonian of the corresponding Jacobi's metric).

The Legendre transformation of a function of the form \eqref{MNJ}
in the presence of constraints \eqref{psi} is given in \cite{FJ}
(see Theorem 2 \cite{FJ} and interchange the role of the tangent
and cotangent bundles of a sphere). As a result, we obtain the
Lagrangian function
$$
L(q,\dot q)=\frac12\langle A\dot q,\dot q\rangle\vert_{S^{n-1}}\,.
$$

Remarkably, after the linear coordinate transformation
$x=\sqrt{A}q$, $L(q,\dot q)$ becomes the Lagrangian $L(x,\dot
x)=\frac12 \langle \dot x,\dot x\rangle$ of the standard metric on
the ellipsoid
$$
E_2^{n-1}=\{x\in\R^n\,\vert\,\langle A^{-1} x,x\rangle=1\}.
$$

Recall that the Reeb flow on $M_0$ can be seen as a time
reparametrization of the original Hamiltonian flow (see Remark
4.1). We can summarize the consideration above in the following
statement.

\begin{prop}[\cite{Ju}]
Under the time substitution
$$
dt=d\tau/2\langle A^{-1}q,q\rangle
$$
and the linear transformation $x=\sqrt{A}q$, the $q$-components of
the trajectories of the system defined by the Hamiltonian function
\eqref{MN} that lye on the zero energy level \eqref{MNo}, become
geodesic lines of the standard metric on the ellipsoid
$E_2^{n-1}$.
\end{prop}

Further interesting examples of transformations related to the
Maupertuis Principle, which map a given integrable system into
another one are given in \cite{Ts}.

\subsection*{Acknowledgments}
I am grateful to Professor Jurdjevic for providing a preprint of
the paper \cite{Ju} and to the referee for valuable remarks. This
research was supported by the Serbian Ministry of Science Project
174020, Geometry and Topology of Manifolds, Classical Mechanics
and Integrable Dynamical Systems.


\begin{thebibliography}{84}

\bibitem{AFKP}
P. Albers, U. Frauenfelder, O. van Koert, G.\, P. Paternain,
\emph{The contact geometry of the restricted 3-body problem},
Comm. Pure Appl. Math. {\bf 65} (2012), no. 2,  229--263.
arXiv:1010.2140v1 [math.SG].


\bibitem{Ar}
{\srrm V. \,I. Arnol\m d,} {\srit Matematicheskie metody
klassichesko{\ji} mehaniki}{\srrm, Moskva, Nauka 1974} (Russian).
English translation:

V.\,I. Arnol'd, \emph{Mathematical methods of classical
mechanics}, Springer-Verlag, 1978.

\bibitem{BM} A. Banyaga, P. Molino,  \emph{G\' eom\' etrie des
formes de contact compl\'etement int\'egrables de type torique},
S\'eminare Gaston Darboux, Montpellier (1991-92), 1-25 (French).
English translation:

A. Banyaga, P. Molino, \emph{Complete Integrability in Contact
Geometry}, Penn State preprint PM 197, 1996.

\bibitem{Bi} {\srrm A. Bilimovi{\cj},} {\srit Racionalna mehanika.
Tom 2 (Mehanika sistema)}{\srrm, Nauchna knjiga, Beograd 1951}
(Serbian).

\bibitem{BKF} {\srrm A. V. Bolsinov, V. V. Kozlov, A. T. Fomenko},
{\srit Princip Mopert{yu}i i geodezicheskie potoki na sferi,
voznika{yu}{shch}ie iz integriruemyh sluchaev dinamiki tverogo
tela}{\srrm, Uspehi Mat. Nauk {\bf 50} (1995), n. 3, 3-32 }
(Russian); English translation:

A. V. Bolsinov, V. V. Kozlov, A. T. Fomenko, \emph{The Maupertuis
principle and geodesic flow on the sphere arising from integrable
cases in the dynamic of a rigid body}, Russian Math. Surv. {\bf
50} (1995).


\bibitem{BJ2} A.\,V. Bolsinov, B. Jovanovi\' c,
\emph{Magnetic Flows on Homogeneous Spaces}, Com. Mat. Helv.,
\textbf{83} (2008), no. 3, 679–700, arXiv: math-ph/0609005.

\bibitem{Bo} C. P. Boyer, \emph{Completely integrable contact
Hamiltonian systems and toric contact structures on $S^2\times
S^3$}, SIGMA {\bf 7} (2011), 058, 22 pages, arXiv: 1101.5587
[math.SG]

\bibitem{CMPR} H. Cendra, J.\, E. Marsden, S. Pekarsky, T.\, S. Ratiu,
\emph{Variational principles for Lie-Poisson and
Hamilton--Poincar\'e equations}, Moscow Math. J. {\bf 3} (2003),
no. 3, 833-867.



\bibitem{CMP}
G. Contreras, L. Macarini, G.P. Paternain, \emph{Periodic orbits
for exact magnetic flows on surfaces}, Int. Math. Res. Not. {\bf
8} (2004) 361-–387.

\bibitem{DM}
{\srrm V. Dragovi{\cj}, D. Milinkovi{\cj}}, {\srit Analiza na
mnogostrukostima}{\srrm, Matematichki fakultet, Beograd}
(Serbian).


\bibitem{FJ}
Y. Fedorov, B. Jovanovi\'c, \emph{Hamiltonization of the
Generalized Veselova LR System}, Reg. Chaot. Dyn. {\bf 14} (2009),
no. 4-5, 495--505.

\bibitem{Gr} P. A. Griffits, \emph{Exterior differential systems and the
calculus of variations}, Progress in Mathematics, 25.
Birkh\"auser, Boston, Mass., 1983.

\bibitem{HZ} H. Hofer, E. Zehnder, \emph{Symplectic Invariants and
Hamiltonian Dynamics}, Birkh\"auser 1994.

\bibitem{Jo} B. Jovanovi\' c, \emph{Noncommutative integrability and action angle variables in contact
geometry}, to appear in J. Sympl. Geometry, arXiv:1103.3611
[math.SG]

\bibitem{Ju} V. Jurdjevic, \emph{Optimal control on Lie groups and
integrable Hamiltonian systems},  Regul. Chaotic Dyn.  {\bf 16}
(2011), no. 5, 514-–535.

\bibitem{KT} B. Khesin, S. Tabachnikov, \emph{Contact complete
integrability, Regular and Chaotic Dynamics}, Special Issue:
Valery Vasilievich Kozlov – 60 (2010), arXiv:0910.0375 [math.SG].

\bibitem{Knorr1}  {H. Kn\"orrer},
\emph{Geodesics on quadrics and a mechanical problem of
C.Neumann}, {J. Reine Angew. Math.} {\bf 334} (1982), 69--78.

\bibitem{Kob} S. Kobayashi,
\emph{Principal fibre bundles with the $1$-dimensional toroidal
group} Tohoku Math. J. (2){\bf 8} (1956), 29-–45.

\bibitem{Ko} {\srrm V. V. Kozlov,} {\srit Variacionnoe ischislenie v celom i klassicheskaya mehanika}{\srrm,
Uspehi Mat. Nauk {\bf 40} (1982), n. 2(242), 33–-60} (Russian).

\bibitem{LM} P. Libermann, C. Marle, \emph{Symplectic Geometry,
Analytical Mechanics}, Riedel, Dordrecht, 1987.

\bibitem{MR} J. E. Marsden, T. S. Ratiu, \emph{Introduction to
mechanics and symmetry}, 2nd edition, Springer 1999.

\bibitem{MF}
{\srrm A.\,S. Mishchenko, A.\,T. Fomenko,} {\srit
Obobshchenny{\ji} metod Liuvill{\ja} integrirovani{\ja}
gamiltonovyh sistem}{\srrm, Funkc. analiz i ego prilozh.
\textbf{12}(2) (1978), 46--56} (Russian); English translation:

A.\,S. Mishchenko, A.\,T. Fomenko, \emph{Generalized Liouville
method of integration of Hamiltonian systems}. Funkts. Anal.
Prilozh. {\bf 12}, No.2, 46-56  (1978)  Funct. Anal. Appl. {\bf
12}, 113--121  (1978)

\bibitem{Mo} J. Moser, \emph{Regularization of Kepler's Problem and the
Averaging Method on a Manifold}, Comm. Pure and Appl. Math. {\bf
23} (1970) 606--636.

\bibitem{Moser} {J. Moser},
\emph{Geometry of quadric and spectral theory}, In: Chern
Symposium 1979, Berlin--Heidelberg--New York, 147--188, 1980.

\bibitem{N}
{\srrm N.\,N. Nehoroshev,} {\srit Peremennye de{\ji}stvie--ugol i
ih obobshcheni{\ja}}{\srrm, Tr. Mosk. Mat. O.-va. \textbf{26}
(1972), 181--198} (Russian). English translation:

N.\, N. Nehoroshev, \emph{Action-angle variables and their
generalization}, {Trans. Mosc. Math. Soc.} {\bf 26},  180--198
(1972).

\bibitem{No} {\srrm S. P. Novikov,} {\srit Gamil\m tonov formalizm
i mnogoznachny{\ji}    analog  teorii  Morsa}{\srrm, Uspehi Mat.
Nauk {\bf 37} (1982), n. 5(227), 3–-49} (Russian).

\bibitem{P} H. Poencar\'e, \emph{Les m\'ethodes nouvelles de la m\'echanique c\'eleste III. Invariant int\'egraux.
Solutions p\'eriodiques du deuxieme genre. Solutions doublement
asymptotiques}, Paris. Gauthier-Villars, 1899 (French).

\bibitem{Po}
{\srrm L.\, S. Polak (ed.),} {\srit Variacionnye principy mehaniki
(Sbornik state{\ji} klassikov nauki)}{\srrm, Fizmatgiz, Moskva
(1959)} (Russian).

\bibitem{Ra}
P. H. Rabinowitz, \emph{Periodic solutions of Hamiltonian
systems}, Comm. Pure. Appl. Math. {\bf 31} (1978), 157–-184.

\bibitem{RBB} G. Romano, R. Barretta, A. Barretta, \emph{On
Maupertuis principle in dynamics}, Rep. Math. Phys. {\bf 63}
(2009) No. 3, 331--346.

\bibitem{Ta} A. I. Taimanov, \emph{Periodic magnetic geodesics on almost every energy level via variational methods},
Regular and Chaotic Dynamics {\bf 15} (2010), 598-605,
arXiv:1001.2677

\bibitem{Ts} A. V. Tsiganov,
\emph{The Maupertuis principle and canonical transformations of
the extended phase space}, J. Nonlinear Math. Phys. {\bf 8}
(2001), no. 1, 157–-182.

\bibitem{Tu} G. M. Tuynman, \emph{Un principle variationnel pour les vari\'et\'es
symplectiques}, C. R. Acad. Sci. Paris S\'er. I Math. {\bf 326}
(1998), no. 3, 339-342 (French).

\bibitem{W1} A. Weinstein, \emph{Bifurcations and Hamiltonian's
principle}, Math. Z. {\bf 159} (1978) 235-248.

\bibitem{We} A. Weinstein, \emph{On the hypotheses of Rabinowitz' periodic
orbit theorems}, J. Differential Equations {\bf 33} (1979), no. 3,
353–-358.

\bibitem{Wh}
E.T. Whittaker, \emph{A treatise on the analytic dynamics of
particles and rigid bodies}, Cambridge, 1904.

\bibitem{Y}
T. Yamazaki, \emph{A construction of K-contact manifolds by a
fiber join}, Tohoku Math. J. {\bf 51} (1999) 433-–446.

\end{thebibliography}
\end{document}